\documentclass[aps,pra,groupedaddress,superscriptaddress,twocolumn,showpacs,letterpaper,nofootinbib]{revtex4-1}

\usepackage{amsmath}
\usepackage{amsfonts}
\usepackage{amsthm}
\usepackage{bm} 
\usepackage{subfigure}
\usepackage{graphicx}

\newtheorem{theorem}{Theorem}

\newcommand{\outlinenote}[1]{}

\providecommand{\abs}[1]{\lvert#1\rvert}

\providecommand{\bra}[1]{\ensuremath{\left\langle{#1}\right|}}
\providecommand{\ket}[1]{\ensuremath{\left|{#1}\right\rangle}}

\providecommand{\braopket}[3]{\ensuremath{\left\langle{#1}|{#2}|{#3}\right\rangle}}
\providecommand{\hilb}[1]{\ensuremath{\mathcal{H}_{#1}}}

\providecommand{\hilbtwo}[2]{\ensuremath{\hilb{#1}\otimes\hilb{#2}}}
\providecommand{\rank}[1]{\ensuremath{\textrm{rank}(#1)}}

\providecommand{\diag}[1]{\ensuremath{\textrm{diag}\{ #1 \}}}

\newcommand{\linop}[2]{\ensuremath{\mathcal{L}(#1,#2)}}
\newcommand{\tE}{\hat{E}}
\newcommand{\tF}{\hat{F}^T}
\newcommand{\tU}{\hat{U}}
\newcommand{\tp}{\hat{\psi}}

\long\def\ca#1\cb{} 

\def\HM{{\mathcal H}}
\def\ad{^\dagger }
\def\al{\alpha }

\def\ot{\otimes }
\def\ra{\rightarrow }

\begin{document}

\title{Entanglement requirements for implementing bipartite unitary operations}

\author{Dan Stahlke}
\email[Electronic address:]{dstahlke@gmail.com}

\author{Robert B. Griffiths}
\email[Electronic address:]{rgrif@andrew.cmu.edu}
\affiliation{Department of Physics, Carnegie Mellon University, Pittsburgh,
Pennsylvania 15213, USA}

\date{July 12, 2011}

\begin{abstract}
  We prove, using a new method based on map-state duality, lower bounds on
  entanglement resources needed to deterministically implement a bipartite
  unitary using separable (SEP) operations, which include LOCC (local
  operations and classical communication) as a particular case.  It is known
  that the Schmidt rank of an entangled pure state resource cannot be less than
  the Schmidt rank of the unitary.  We prove that if these ranks are equal the
  resource must be uniformly (maximally) entangled: equal nonzero Schmidt
  coefficients.  Higher rank resources can have less entanglement: we have
  found numerical examples of Schmidt rank 2 unitaries which can be
  deterministically implemented, by either SEP or LOCC, using an entangled
  resource of two qutrits with less than one ebit of entanglement.
\end{abstract}


\maketitle

\section{Introduction}

\outlinenote{Explain LOCC and resource}

It is possible to carry out nonlocal quantum operations on multipartite
systems using only local quantum operations and classical communications
(LOCC) provided that the parties involved have access to a suitable entangled
state, referred to as a resource.  Given a large enough resource it is always
possible to use teleportation to send all inputs to one party, who performs
the operation and then distributes the results to the other parties using
teleportation.  In some cases it is possible to perform a nonlocal operation
with less entanglement than is required by
teleportation~\cite{PhysRevA.62.052317,PhysRevA.65.032312,PhysRevA.81.062315,PhysRevA.81.062316,PhysRevA.82.022313,PhysRevA.78.014305}.
The question then arises as to how much entanglement is really necessary in
order to implement a given nonlocal operation.

\outlinenote{What is known for nonlocal unitaries}

Our first result, that the Schmidt rank of the resource must be at least
as great as that of the unitary [Theorem~1(a)], follows
rather immediately from the fact that it is a separable (SEP) operation.
This is analogous to the result given in~\cite{PhysRevLett.89.057901}
in which probabilistic (i.e.~SLOCC) implementations are considered.
Since SEP is contained in SLOCC, our Theorem~1(a) can be seen as a consequence
of the result in~\cite{PhysRevLett.89.057901}, however we provide an
independent proof along the way to our main result.

In contrast to the probabilistic case, the deterministic
implementation of a unitary is only possible if the state meets certain
entanglement requirements.
For one thing, the entanglement of the resource must be at least as great as
the entangling power of the unitary since entanglement cannot increase under
SEP~\cite{PhysRevA.78.020304}.
It has been shown that any deterministic controlled-unitary operator on two
qubits implemented with bipartite LOCC using a resource of two entangled qubits
necessarily requires a maximally entangled resource~\cite{soeda10}.  Our paper
takes a different approach to the problem, using SEP,
and provides a proof applicable to general unitaries of arbitrary dimension.
We show that if the resource has Schmidt rank equal to that of the unitary, the
resource must be uniformly entangled in the sense that all its nonzero Schmidt
coefficients are the same [Theorem~1(b)].
These same restrictions apply to LOCC, as it is a particular case of SEP.

\outlinenote{Counter example for larger rank}

It is not hard to see that if the Schmidt rank of the resource is greater than
the Schmidt rank of the unitary, then the resource need not be uniformly
entangled (e.g.\ a larger rank resource that is majorized by a smaller rank
maximally entangled state).
We have found that it is in fact possible for such a larger rank resource to
have less entanglement than would be required for a resource of Schmidt rank
equal to that of the unitary.
We have found examples of protocols in both SEP and LOCC which deterministically
implement a controlled phase operation using less than one ebit of entanglement.
In this case the unitary has Schmidt rank two and the resource has Schmidt
rank three.
Although the nonlocal unitary protocol given in~\cite{PhysRevLett.86.544} can
with certain probability consume less than one ebit of entanglement~\footnote{
Although the protocol given in~\cite{PhysRevLett.86.544} is deterministic in
the sense that it always succeeds in a finite number of steps, it is
probabilistic in the amount of entanglement required.  For any nontrivial
unitary there is a chance that the protocol requires usage of the
$\ket{\psi_{\alpha_2}}$ state, which has one ebit of entanglement.  Thus, if
the protocol only has access to a state with less than one ebit of entanglement
there is a nonzero probability that the protocol cannot be carried out
successfully.
}, we believe
that ours is the first example of carrying out such a protocol
deterministically using less than one ebit of entanglement.

\outlinenote{A brief description of what is in each section}

The remainder of this article is organized as follows. Section~\ref{sec:SEP}
sets up the problem of bipartite deterministic implementations of unitary
operators using SEP\@.
Section~\ref{sec:duality} provides the requisite background regarding map-state
duality~\cite{ZcBn04,BnZc06} and atemporal
diagrams~\cite{PhysRevA.73.052309,Stedman200909,Cvitanovic200807}.
Our main result is proved in Sec.~\ref{sec:main} using what we believe to be a
new method based on the use of map-state duality.
In Sec.~\ref{sec:4} we consider the case of a resource of larger Schmidt rank.
There is a brief conclusion in Sec.~\ref{sec:conclusions}.
An appendix details the implementation of a controlled unitary using a
qutrit resource state of less than one ebit of entanglement.

\section{Nonlocal Unitaries Via Separable operations}
\label{sec:SEP}

We are interested in carrying out a bipartite unitary map
$U:\hilb{A}\ot\hilb{B}\ra \hilb{\bar A}\ot\hilb{\bar B}$, using as a resource
an entangled state $\ket{\psi}$ on two ancillary systems $\hilb{a}$ and
$\hilb{b}$, by means of a separable operation $\{E_k\ot F_k\}$, $k=1, 2,
\ldots$.  Here $E_k:\hilb{A}\ot\hilb{a}\ra \hilb{\bar A}$ and
$F_k:\hilb{B}\ot\hilb{b}\ra \hilb{\bar B}$ together form a product Kraus
operator. For $U$ to be unitary
it is necessary that the dimensions of the Hilbert spaces satisfy $d_Ad_B =
d_{\bar A}d_{\bar B}$, but we do not require that $d_A = d_{\bar A}$ or $d_A =
d_B$.  The separable operation must satisfy the usual closure
condition~\cite{RevModPhys.81.865}
\begin{equation}
  \label{eqn1}
	\sum_k (E_k \otimes F_k)^\dagger (E_k \otimes F_k) =
		I_A \otimes I_a \otimes I_b \otimes I_B
\end{equation}
which is depicted in Fig.~\ref{fig1}(a).

In addition, for $\ket{\Phi}$ any pure input state on $\hilb{A}\ot\hilb{B}$, the
outcome of the operation will be a pure state
\begin{align}
  \label{eqn2}
&U \Bigl(\ket{\Phi} \bra{\Phi}\Bigr) U^\dagger = 
\notag\\
 &\sum_k \Bigl(E_k \otimes F_k\Bigr) 
\Bigl(\ket{\Phi}\bra{\Phi}\otimes\ket{\psi} \bra{\psi}\Bigr) 
\Bigl(E_k \otimes F_k\Bigr)^\dagger
\end{align}
on $\hilb{\bar A}\ot\hilb{\bar B}$.  Since the protocol is assumed to be
deterministic, every term on the right side is proportional to the same pure
state and it must be the case that
\begin{equation}
  \label{eqn3}
 (E_k\ot F_k) \ket{\psi} = \al_k U,
\end{equation}
with $\al_k$ some complex number.  Note that both sides of \eqref{eqn3} are
operators acting on $\HM_A\ot\HM_B$; Fig.~\ref{fig2}(a) will help interpreting it
correctly.

The resource $\ket{\psi}$ is assumed to have a Schmidt rank of $D_\psi$, which
means it can be written in the form
\begin{equation}
  \label{eqn4}
	\ket{\psi} = \sum_{i=1}^{D_\psi} \lambda_i \ket{a_i} \otimes \ket{b_i}.
\end{equation}
for suitable orthonormal bases $\{\ket{a_i}\}$ and $\{\ket{b_i}\}$ of $\HM_a$
and $\HM_b$, with Schmidt coefficients $\lambda_i>0$ for $i\leq D_\psi$.

Similarly, the bipartite operator $U$ is assumed to have a Schmidt rank
of $D_U$, meaning that it can be written in the form~\cite{0305-4470-36-39-309}
\begin{equation}
  \label{eqn5}
	U = \sum_{i=1}^{D_U} \mu_i \mathcal{A}_i \otimes \mathcal{B}_i,
\end{equation}
where $\{\mathcal{A}_i\}$ and $\{\mathcal{B}_i\}$ are bases of the
operator spaces $\linop{\hilb{A}, \hilb{\bar A}}$ and
$\linop{\hilb{B}, \hilb{\bar B}}$, orthonormal under
the Frobenius (Hilbert-Schmidt) inner product, and $\mu_i>0$ for
$i \leq D_U$.
Equivalently, $D_U$ is the minimum number of terms needed in order
to write $U$ in the form $\sum \mathcal{C}_i \otimes \mathcal{D}_i$,
without requiring $\mathcal{C}_i$ or $\mathcal{D}_i$ to be from an
orthonormal basis.

\section{Map-State Duality and Diagrams}
\label{sec:duality}

\textit{Map-state duality}~\cite{ZcBn04,BnZc06}
plays a central role in the proof that will follow.  This is a general concept
that is sometimes referred to as reshaping or a partial transpose~\cite{ZcBn04}
and in a specific manifestation is known as the Jamio{\l}kowski or sometimes
the Choi-Jamio{\l}kowski isomorphism.
States and maps are considered to both be tensors, and when a
choice of orthonormal basis is fixed there is a natural linear relation between
bras and kets (i.e.\ $\ket{i} \leftrightarrow \bra{i}$ for all basis vectors
$\ket{i}$)~\footnote{
It is also possible to formulate map-state duality in a basis independent
manner~\cite{0305-4470-36-39-309}, however this is not necessary for
the present work.
}.

With this identification between bras and kets in place, the
bipartite state $\ket{\psi}$ on the Hilbert space $\hilbtwo{a}{b}$ can be
identified with the linear map $\psi':\hilb{b} \to \hilb{a}$ obtained by
turning kets into bras on the $\hilb{b}$ space:
\begin{equation}
\ket{\psi} = \sum_{ij} \psi_{ij} \ket{a_i} \ot \ket{b_j} \to
\psi' = \sum_{ij} \psi_{ij} \ket{a_i}\bra{b_j}
\label{eqn6}
\end{equation}
Similarly, the operators $U$, $E_k$, and $F_k$, give rise to
$U':\hilbtwo{B}{\bar B} \to \hilbtwo{A}{\bar A}$
(by turning bras into kets on $\hilb{A}$ and kets into bras on $\hilb{\bar B}$),
$E'_k:\hilb{a} \to \hilbtwo{A}{\bar A}$
(by turning bras into kets on $\hilb{A}$), and
$F'^T_k:\hilbtwo{B}{\bar B} \to \hilb{b}$
(by turning bras into kets on $\hilb{b}$ and kets into bras on $\hilb{\bar B}$),
\begin{align}
{U'}   &= \sum_ {ijmn} \braopket{\bar A_j, \bar B_n}{U}{A_i, B_m} \ket{A_i, \bar A_j} \bra{B_m, \bar B_n}, \notag \\
{E'_k} &= \sum_{ijm} \braopket{\bar A_j}{E_k}{A_i,a_m} \ket{A_i, \bar A_j} \bra{a_m}, \notag \\
F'^T_k &= \sum_{ijm} \braopket{\bar B_j}{F_k}{B_i,b_m} \ket{b_m} \bra{B_i, \bar B_j}.
\label{eqn7}
\end{align}
In the case of these three operators, \textit{map-map duality} may be
a more precise term, however we will use map-state duality to refer to any such
partial transpose.
The primed operator for $F_k$ is denoted as $F'^T_k$ in order to draw attention
to the fact that its domain and range are swapped in comparison to $E'_k$.

The equations introduced so far make use of six distinct Hilbert spaces and tensors
of various rank.  In such situations the underlying structure of equations
can be somewhat hidden when expressed using Dirac notation.  Abstract index notation is more
transparent but can become unwieldy.  For this reason we provide atemporal
diagrams, similar to those found in~\cite{PhysRevA.73.052309},
which should aid the reader in following the arguments in the text.

Operators are designated by squares or rectangular boxes.  As a matter of style, the
state $\ket{\psi}$ and its corresponding operator $\psi'$ will be represented
as a circle instead of a square.  Lines between these boxes represent tensor
contraction, and these lines are labeled by the Hilbert spaces which they
correspond to.  Open lines on the left of a diagram represent the input to
the total linear operator defined by the diagram, and open lines on the right
represent outputs.  Putting the inputs on the left means that operators are to
be applied in a left-to-right manner, opposite to how algebraic equations are
interpreted.  As has been so far described, our diagrams are to be interpreted
in exactly the same way as traditional quantum circuits as used for example
in Nielsen and Chuang~\cite{NielsenChuang200009}.
The primary difference between our diagrams and
traditional circuits is that in the latter the horizontal direction is understood
to represent the passage of time whereas our diagrams make no reference to time.
The presence of a summation symbol has the obvious meaning: the linear operator
depicted in the diagram denotes the terms of a series.
The trace or partial trace operation is just a special case of tensor contraction
and is denoted by joining the relevant spaces with a line.
The identity operator is represented by a line.
With minor changes in style our diagrams are equivalent to the atemporal
diagrams of~\cite{PhysRevA.73.052309}, and resemble other such
schemes~\cite{Stedman200909,Cvitanovic200807}.

\begin{figure}[t]
	\centering
	\includegraphics{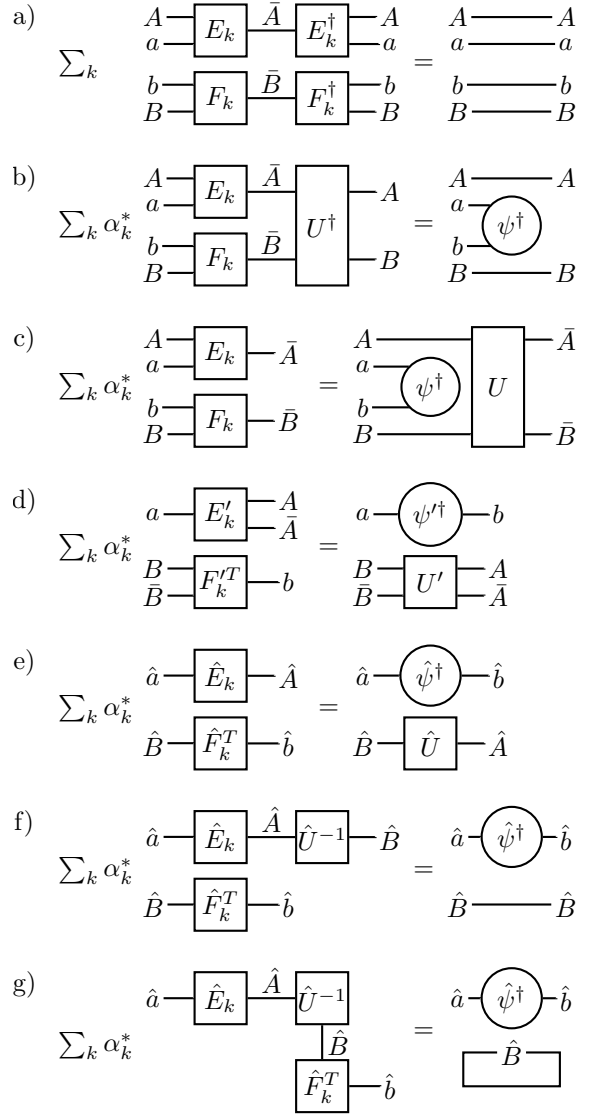}
	\caption{
		Atemporal diagrams, explained in Sec.~\ref{sec:duality}.
		(a) Closure condition, \eqref{eqn1}.
		(b) Apply $\bra{\psi}$ and simplify using the adjoint of Fig.~\ref{fig2}(a) to get \eqref{eqn9}.
		(c) Multiply on the right by $U$ to get \eqref{eqn10}.
		(d) Apply map-state duality to get \eqref{eqn11}.
		(e) Restrict spaces to supports and ranges of operators to get \eqref{eqn12}.
		(f) Multiply by $\tU^{-1}$.
		(g) Trace over $\hilb{\hat B}$ to get \eqref{eqn13}.
	}
	\label{fig1}
\end{figure}

\begin{figure}[t]
	\centering
	\includegraphics{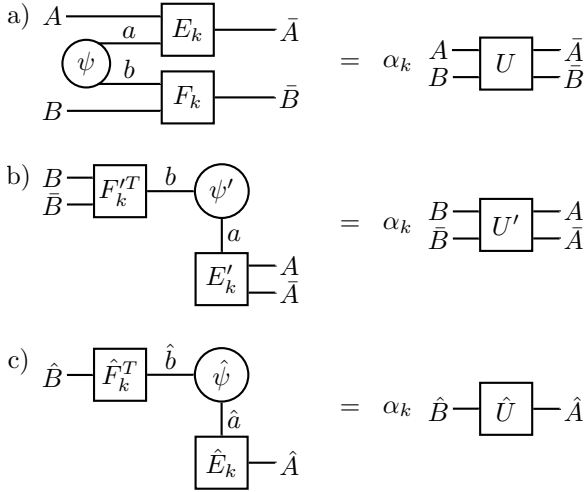}
	\caption{
		(a) Deterministic unitary operation, \eqref{eqn3}.
		(b) Apply map-state duality to get \eqref{eqn8}.
		(c) Restrict spaces to supports and ranges of operators to get \eqref{eqn14}.
	}
	\label{fig2}
\end{figure}

\section{Entanglement Requirements}
\label{sec:main}


Our main result is the following:

\outlinenote{The theorem}

\begin{theorem}
\label{thm:max_entangled}
Suppose that a unitary operator $U$ is implemented deterministically by a separable operation
that makes use of the pure state entanglement resource $\ket{\psi}$
[i.e.\ suppose that \eqref{eqn1} and \eqref{eqn3} hold].
Then

\textrm{(a)} The Schmidt rank $D_\psi$ of $\ket{\psi}$ is greater than or
equal to the Schmidt rank $D_U$ of $U$.

\textrm{(b)} If the Schmidt ranks are equal, $D_U=D_\psi$, then $\ket{\psi}$ must be a
uniformly (maximally) entangled state: all the nonzero Schmidt coefficients
are the same. 
\end{theorem}

\begin{proof}[Proof of (a)]
Making use of map-state duality and the operators defined in \eqref{eqn6} and \eqref{eqn7},
equation \eqref{eqn3} [Fig.~\ref{fig2}(a)] can be rewritten as [Fig.~\ref{fig2}(b)]
\begin{equation}
E'_k \psi' F'^T_k = \al_k U'.
\label{eqn8}
\end{equation}
Since the rank of a product of linear operators is at most the smallest of the ranks of the individual
operators, it follows that $\rank{\psi'} \ge \rank{U'}$.
The rank of an operator is equal to the number of its nonzero singular values.
Since the Schmidt decompositions \eqref{eqn4} and \eqref{eqn5} are essentially singular value
decompositions of $\psi'$ and $U'$, it is apparent that $\rank{\psi'} = D_\psi$ and $\rank{U'} = D_U$
and the inequality becomes $D_\psi \ge D_U$.  Part (a) is proved.
\end{proof}

\begin{proof}[Proof of (b)]
Apply the closure condition \eqref{eqn1} to $\bra{\psi}$
and use the adjoint of \eqref{eqn3} to obtain
\begin{equation}
  \label{eqn9}
\sum_k \al_k^* U\ad (E_k\ot F_k) = \bra{\psi} \ot I_{AB}.
\end{equation}
as shown in Fig.~\ref{fig1}(b). 
Next, multiply both sides on the left by $U$ to arrive at
\begin{equation}
  \label{eqn10}
	\sum_k \alpha_k^* (E_k \otimes F_k) = \bra{\psi} \otimes U,
\end{equation}
as shown in Fig.~\ref{fig1}(c).  Making use of map-state duality gives [Fig.~\ref{fig1}(d)]
\begin{equation}
  \label{eqn11}
	\sum_k \alpha_k^* (E'_k \otimes F'^T_k) = \psi'^\dagger \otimes U'.
\end{equation}

The map $U'$ may in general have rank less than the dimension of $\HM_B\ot\HM_{\bar B}$ or
$\HM_A\ot\HM_{\bar A}$ (which need not be equal to each other).  In this case
it will be useful to denote by $\HM_{\hat B}$ the subspace of $\HM_B\ot\HM_{\bar
  B}$ which forms the \emph{support} (or co-image or row space) of $U'$, the
orthogonal complement of its kernel (null space), and by $\HM_{\hat A}$ the
subspace of $\HM_A\ot\HM_{\bar A}$ that forms the \emph{range} (or image) of
$U'$.  Each of these subspaces has a dimension equal to $D_U$, and $U'$ is a
nonsingular (invertible) linear map of $\HM_{\hat B}$ onto $\HM_{\hat A}$,
which we hereafter denote by $\hat U$.  In the same way one can introduce
subspaces $\HM_{\hat b}$ and $\HM_{\hat a}$ of $\HM_b$ and $\HM_a$ which form
the support and range of $\psi'$, and define $\hat\psi$ to be the corresponding
nonsingular map of rank $D_\psi$ from $\HM_{\hat b}$ to $\HM_{\hat a}$.  Next,
$\hat E_k$ is $E'_k$ with its domain restricted to $\HM_{\hat a}$, which can be
strictly smaller than the support of $E'_k$, and with its range restricted to
$\HM_{\hat A}$, which could be smaller than the image of $E'_k$.  Finally,
$\hat{F}_k^T$ is $F'^T_k$ regarded as a map from $\HM_B\ot\HM_{\bar B}$ to $\HM_b$,
but with domain and range restricted to $\HM_{\hat B}$ and $\HM_{\hat b}$,
respectively~\footnote{
It is significant that we define $\hat{F}_k^T$ as $F'^T_k$ restricted to subspaces.
In general it is not the case that $\hat{F}_k$ is equal to $F'_k$ restricted to subspaces.
}.

The result of restricting \eqref{eqn11} to the subspaces just defined is
\begin{equation}
  \label{eqn12}
	\sum_k \alpha_k^* (\tE_k \otimes \tF_k) = \tp^\dagger \otimes \tU,
\end{equation}
corresponding to Fig.~\ref{fig1}(e). Multiplying on the left by $\tU^{-1}$ and tracing
over $\hilb{\hat{B}}$ gives 
\begin{equation}
  \label{eqn13}
\sum_k \alpha_k^* \tF_k \tU^{-1} \tE_k = D_U\tp^\dagger ,
\end{equation}
see Fig.~\ref{fig1}(f) and (g).  Restricting \eqref{eqn8} to subspaces results in 
\begin{equation}
  \label{eqn14}
	\tE_k \tp \tF_k = \alpha_k \tU,
\end{equation}
see Fig.~\ref{fig2}(c).
Here we have restricted the spaces over which matrix multiplications are being
performed ($\HM_{\hat b}$ and $\HM_{\hat a}$ instead of $\HM_b$ and $\HM_a$),
however equality is still maintained because the dimensions which have been
eliminated correspond to the zero Schmidt coefficients of $\ket{\psi}$, which
is to say the zero singular values of $\psi'$.

To complete the proof, make use of the assumption $D_\psi=D_U$.
Then $D_\psi$ is also the rank of
$\tE_k$ and $\tF_k$: all four operators in \eqref{eqn14} are full rank.
Taking the inverse of both sides and inserting the result for $\hat U^{-1}$ in
\eqref{eqn13} leads to the result
\begin{equation}
  \label{eqn15}
\sum_k \abs{\alpha_k}^2	\tp^{-1} = \tp^{-1} = D_\psi \tp^\dagger ,
\end{equation}
where  $\sum_k|\al_k|^2=1$  follows from  \eqref{eqn1},  \eqref{eqn3} and  the
normalization of  $\ket{\psi}$.  With $\ket{\psi}$  in Schmidt form,  $\tp$ is
diagonal,  so $\tp  =  I/\sqrt{D_\psi}$.  Therefore  all  the nonzero  Schmidt
coefficients  of $\ket{\psi}$  are equal  to $1/\sqrt{D_\psi}$.

\end{proof}


\section{Larger Rank Resource}
\label{sec:4}

\outlinenote{Issue of resource of greater rank}

We have proved that a resource that is of the smallest viable Schmidt rank must
be maximally entangled, but it is also possible to use a resource that is of
higher Schmidt rank that is not maximally entangled.
For one thing, if such a state meets an appropriate majorization criterion it
can be deterministically transformed into a maximally entangled
state~\cite{PhysRevLett.83.436}.
In this case the larger rank initial resource would have greater entanglement
than would be required if the smaller maximally entangled state had been used
in the first place.
There is however the possibility that some protocol could be devised to use a
resource of larger Schmidt rank that has less entanglement than the maximally
entangled state of smaller rank.

\outlinenote{Numerical solution in SEP}

In fact, we have numerically found examples of such constructions in both SEP
and LOCC\@.  One solution in SEP uses a resource state $\ket{\psi} = \sqrt{0.81}\ket{00}
+ \sqrt{0.095}(\ket{11}+\ket{22})$ on two qutrits to implement the two qubit controlled
unitary operator $U = \textrm{diag}\{1, 1, 1, e^{i \phi}\}$ with
$\phi = 2 \textrm{cos}^{-1}(35/36)$.  We have verified this to be an exact
solution using a computer algebra system.  This resource
constitutes less than one ebit of entanglement: the Von Neumann entropy is
approximately $0.89$ ebits.  Since entropy cannot increase under
SEP~\cite{PhysRevA.78.020304} it is necessary for the resource that is consumed
to have greater entanglement than the entangling capacity of the unitary being
implemented.  The entangling capacity of this unitary is shown
in~\cite{PhysRevA.67.012306} to be approximately $0.23$ ebits.  Since this is
much less than the $0.89$ ebits that we use, there remains the possibility that
a different construction or an even larger rank resource could potentially
lower the entanglement cost further.

\outlinenote{Numerical solution in LOCC}

We also found an LOCC protocol which, though less efficient than the
SEP construction
just described, allows one to carry out a  bipartite unitary
deterministically using a resource with less than
one ebit of entanglement.
The resource in this case is
$\ket{\psi} = \sqrt{0.8}\ket{00} + \sqrt{0.1}(\ket{11}+\ket{22})$
and the unitary implemented is
$U = \textrm{diag}\{1, 1, 1, e^{i \phi}\}$ with $\phi = 0.08 \pi$.
The Von Neumann entropy of this resource is approximately $0.92$ ebits,
and this is a four round protocol (Alice, Bob, Alice, Bob).

The constructions described above are instances of a more general continuous family of
solutions that we have found, covering a range of controlled phase operations.
As should be expected, a larger phase $\phi$ requires a larger entanglement resource.
In both the SEP and the LOCC case only certain classes of solutions were
searched for, so it is possible that a more thorough search would provide
more efficient protocols.
The details of our SEP construction are presented in Appendix~\ref{sec:sepappendix}.
Our LOCC construction consists of a long list of Kraus operators in
numerical form, which is available upon request.


\section{Conclusion}
\label{sec:conclusions}

We have shown that a unitary operator of Schmidt rank $D$ implemented as a
bipartite separable operation requires an entanglement resource of Schmidt
rank at least $D$.  If the Schmidt rank of the resource is exactly equal to
$D$, the resource must be uniformly (maximally) entangled with equal nonzero
Schmidt coefficients.   These restrictions apply also to LOCC, which is a
subset of SEP\@.  The proof uses map-state duality in a way which has not (so
far as we know) been previously applied to problems of this type, so might
have other interesting applications.

Numerical results show that the amount of entanglement required for the
resource can be lowered by using a resource of Schmidt rank larger than $D$.
A four round LOCC protocol has been found which uses a two-qutrit resource state with
less than one ebit of entanglement to implement a bipartite controlled phase
gate (albeit with a small phase).

Although some large classes of unitaries are known to have implementations
in LOCC using resources having the minimal Schmidt rank required by
Theorem~1(a)~\cite{PhysRevA.62.052317,PhysRevA.65.032312,PhysRevA.81.062315,PhysRevA.81.062316,PhysRevA.78.014305},
it is not known whether such minimal-rank implementations are possible for all
unitaries.
Given a unitary of Schmidt rank $D_U$ it is always possible to find a
collection of operators $\{E_k\otimes F_k\}$ such that \eqref{eqn9} and
\eqref{eqn3} are satisfied with a resource of Schmidt rank $D_\psi=D_U$.
But it is not known if there is a separable operation satisfying both
\eqref{eqn1} and \eqref{eqn3}.  Consequently, it is possible that some
unitaries may require a resource of greater rank than the lower bound given in
Theorem~1(a).
Even if such a minimal rank solution is always possible in SEP, it still might
not be possible in LOCC.
This stands in contrast to the case of SLOCC where it is known that any unitary
can be implemented using a state of Schmidt rank equal to that of the
unitary~\cite{PhysRevLett.89.057901}.


\begin{acknowledgments}

We thank Vlad Gheorghiu for his comments. The
research reported here was supported in part by the National Science
Foundation through Grant No.\ 0757251.
\end{acknowledgments}


\appendix
\section{Less than one ebit in SEP}
\label{sec:sepappendix}

We performed a numerical search for solutions to~\eqref{eqn1} and~\eqref{eqn3} with
the resource and unitary taking the forms
\begin{align}
\ket{\psi} &= \sqrt{c_0}\ket{00} + \sqrt{(1-c_0)/2}(\ket{11}+\ket{22}), \\
U &= \textrm{diag}\{ 1, 1, 1, e^{i \theta} \}.
\end{align}
In this case the unitary $U$ is Schmidt rank 2 and the resource is Schmidt
rank 3, so the spaces $\hilb{a}$ and $\hilb{b}$ are each 3 dimensional.
In order to reduce the search space we looked for operators $\{E_k\}$ and
$\{F_k\}$ of the form
\begin{align}
E_k = E_* S_k \textrm{\ and\ } F_k = F_* T_k
\end{align}
where $S_k : \hilb{a} \to \hilb{c}$, $T_k : \hilb{b} \to \hilb{c}$ with
$\hilb{c}$ being a two dimensional space, and
\begin{align}
E_* &=
\left(\begin{array}{rr} 1 & 0 \\ 0 & 0 \end{array}\right)_{\bar{A} A} \otimes \bra{0}_c +
\left(\begin{array}{rr} 0 & 0 \\ 0 & 1 \end{array}\right)_{\bar{A} A} \otimes \bra{1}_c,
\\
F_* &=
\left(\begin{array}{rr} 1 & 0 \\ 0 & 1            \end{array}\right)_{\bar{B} B} \otimes \bra{0}_c +
\left(\begin{array}{rr} 1 & 0 \\ 0 & e^{i \theta} \end{array}\right)_{\bar{B} B} \otimes \bra{1}_c.
\end{align}
It is possible to take advantage of the symmetry of the resource
$\ket{\psi}$ by searching for operator sets of the form
\begin{align}
\{ S_k L^l M^m N^n \} \textrm{\ and\ } \{ T_k L^l M^m N^n \}
\label{eqn:sym_st}
\end{align}
where $l, m, n \in \{0,1\}$ and $L$, $M$, and $N$ are defined by
\begin{align}
L =
	\left(\begin{array}{rrr}
	1 & 0 & 0 \\
	0 & 0 & 1 \\
	0 & 1 & 0 
	\end{array}\right),
\end{align}
$M = \diag{ 1, 1, -1 }$, and $N = \diag{ 1, -1, 1 }$.
There is no loss of generality in this assumption, since if
$(\{S_k\}, \{T_k\})$ gives a solution to~\eqref{eqn1} and~\eqref{eqn3}
then so does
$(\{\frac{1}{\sqrt{8}} S_k L^l M^m N^n\}, \{T_k L^l M^m N^n\})$.
This decreases the number of independent operators (indexed by $k$) that need
to be solved for, and in fact it turns out to be sufficient to consider
only two values of $k$.

Initially we searched for solutions with $\theta=\pi/4$ and $c_0 = 0.6$ which,
although representing more than one ebit of entanglement, is not majorized by a
fully entangled resource of Schmidt rank 2.  Once a solution was found, the
parameters were variated until a value of $c_0$ was reached which represented a
resource of less than one ebit of entanglement.  Further constraints were added
and variations made to simplify the solution and identify relations between the
parameters.  A family of solutions was found of the form \eqref{eqn:sym_st} with
\begin{align}
S_0 &=
\left(\begin{array}{rrr}
p & 1 & -p \\
e^{i \theta/2} & -p & -1
\end{array}\right),
\\
S_1 &=
\left(\begin{array}{rrr}
-1 & 1 & -p \\
-p e^{i \theta/2} & p & 1
\end{array}\right),
\\
T_0 &=
\left(\begin{array}{rrr}
-x-y & * & * \\
(x-y) e^{-i \theta/2} & * & *
\end{array}\right),
\\
T_1 &=
\left(\begin{array}{rrr}
-x+y & * & * \\
(x+y) e^{-i \theta/2} & * & *
\end{array}\right),
\\
p &= \sqrt{\frac{1-s}{s}},
\\
s &= x^2 \left(1-cos\left(\theta/2\right)\right) + 
     y^2 \left(1+cos\left(\theta/2\right)\right),
\end{align}
where the parameters $x$, $y$, $c_0$, and $\theta$ must be solved for
numerically.  The asterisks in $T_0$ and $T_1$ represent parameters that can be
found using the relation $S_k \psi' T^T_k = I/4$.

A sequence of solutions for $x$, $y$, $c_0$, and $\theta$ were fed into an
inverse symbolic calculator of our own design which uses a lookup table to
convert floating point numbers into algebraic expressions.  One of these
solutions produced particularly simple algebraic expressions:
\begin{align}
x &= 9/5, \\
y &= -3/5, \\
c_0 &= 0.81, \\
\theta &= 2 \textrm{arccos}(35/36).
\end{align}
With this algebraic solution in hand, we used the computer algebra package
Sage~\cite{sage} to verify that this indeed represented an exact (not just
approximate to within floating point precision) solution to~\eqref{eqn1}
and~\eqref{eqn3}.



\begin{thebibliography}{21}
\expandafter\ifx\csname natexlab\endcsname\relax\def\natexlab#1{#1}\fi
\expandafter\ifx\csname bibnamefont\endcsname\relax
  \def\bibnamefont#1{#1}\fi
\expandafter\ifx\csname bibfnamefont\endcsname\relax
  \def\bibfnamefont#1{#1}\fi
\expandafter\ifx\csname citenamefont\endcsname\relax
  \def\citenamefont#1{#1}\fi
\expandafter\ifx\csname url\endcsname\relax
  \def\url#1{\texttt{#1}}\fi
\expandafter\ifx\csname urlprefix\endcsname\relax\def\urlprefix{URL }\fi
\providecommand{\bibinfo}[2]{#2}
\providecommand{\eprint}[2][]{\url{#2}}

\bibitem[{\citenamefont{Eisert et~al.}(2000)\citenamefont{Eisert, Jacobs,
  Papadopoulos, and Plenio}}]{PhysRevA.62.052317}
\bibinfo{author}{\bibfnamefont{J.}~\bibnamefont{Eisert}},
  \bibinfo{author}{\bibfnamefont{K.}~\bibnamefont{Jacobs}},
  \bibinfo{author}{\bibfnamefont{P.}~\bibnamefont{Papadopoulos}},
  \bibnamefont{and} \bibinfo{author}{\bibfnamefont{M.~B.}
  \bibnamefont{Plenio}}, \bibinfo{journal}{Phys. Rev. A}
  \textbf{\bibinfo{volume}{62}}, \bibinfo{pages}{052317}
  (\bibinfo{year}{2000}).

\bibitem[{\citenamefont{Reznik et~al.}(2002)\citenamefont{Reznik, Aharonov, and
  Groisman}}]{PhysRevA.65.032312}
\bibinfo{author}{\bibfnamefont{B.}~\bibnamefont{Reznik}},
  \bibinfo{author}{\bibfnamefont{Y.}~\bibnamefont{Aharonov}}, \bibnamefont{and}
  \bibinfo{author}{\bibfnamefont{B.}~\bibnamefont{Groisman}},
  \bibinfo{journal}{Phys. Rev. A} \textbf{\bibinfo{volume}{65}},
  \bibinfo{pages}{032312} (\bibinfo{year}{2002}).

\bibitem[{\citenamefont{Yu et~al.}(2010)\citenamefont{Yu, Griffiths, and
  Cohen}}]{PhysRevA.81.062315}
\bibinfo{author}{\bibfnamefont{L.}~\bibnamefont{Yu}},
  \bibinfo{author}{\bibfnamefont{R.~B.} \bibnamefont{Griffiths}},
  \bibnamefont{and} \bibinfo{author}{\bibfnamefont{S.~M.} \bibnamefont{Cohen}},
  \bibinfo{journal}{Phys. Rev. A} \textbf{\bibinfo{volume}{81}},
  \bibinfo{pages}{062315} (\bibinfo{year}{2010}).

\bibitem[{\citenamefont{Cohen}(2010)}]{PhysRevA.81.062316}
\bibinfo{author}{\bibfnamefont{S.~M.} \bibnamefont{Cohen}},
  \bibinfo{journal}{Phys. Rev. A} \textbf{\bibinfo{volume}{81}},
  \bibinfo{pages}{062316} (\bibinfo{year}{2010}).

\bibitem[{\citenamefont{Gheorghiu et~al.}(2010)\citenamefont{Gheorghiu, Yu, and
  Cohen}}]{PhysRevA.82.022313}
\bibinfo{author}{\bibfnamefont{V.}~\bibnamefont{Gheorghiu}},
  \bibinfo{author}{\bibfnamefont{L.}~\bibnamefont{Yu}}, \bibnamefont{and}
  \bibinfo{author}{\bibfnamefont{S.~M.} \bibnamefont{Cohen}},
  \bibinfo{journal}{Phys. Rev. A} \textbf{\bibinfo{volume}{82}},
  \bibinfo{pages}{022313} (\bibinfo{year}{2010}).

\bibitem[{\citenamefont{Zhao and Wang}(2008)}]{PhysRevA.78.014305}
\bibinfo{author}{\bibfnamefont{N.~B.} \bibnamefont{Zhao}} \bibnamefont{and}
  \bibinfo{author}{\bibfnamefont{A.~M.} \bibnamefont{Wang}},
  \bibinfo{journal}{Phys. Rev. A} \textbf{\bibinfo{volume}{78}},
  \bibinfo{pages}{014305} (\bibinfo{year}{2008}).

\bibitem[{\citenamefont{D\"ur et~al.}(2002)\citenamefont{D\"ur, Vidal, and
  Cirac}}]{PhysRevLett.89.057901}
\bibinfo{author}{\bibfnamefont{W.}~\bibnamefont{D\"ur}},
  \bibinfo{author}{\bibfnamefont{G.}~\bibnamefont{Vidal}}, \bibnamefont{and}
  \bibinfo{author}{\bibfnamefont{J.~I.} \bibnamefont{Cirac}},
  \bibinfo{journal}{Phys. Rev. Lett.} \textbf{\bibinfo{volume}{89}},
  \bibinfo{pages}{057901} (\bibinfo{year}{2002}).

\bibitem[{\citenamefont{Gheorghiu and Griffiths}(2008)}]{PhysRevA.78.020304}
\bibinfo{author}{\bibfnamefont{V.}~\bibnamefont{Gheorghiu}} \bibnamefont{and}
  \bibinfo{author}{\bibfnamefont{R.~B.} \bibnamefont{Griffiths}},
  \bibinfo{journal}{Phys. Rev. A} \textbf{\bibinfo{volume}{78}},
  \bibinfo{pages}{020304} (\bibinfo{year}{2008}).

\bibitem[{\citenamefont{Soeda et~al.}(2010)\citenamefont{Soeda, Turner, and
  Murao}}]{soeda10}
\bibinfo{author}{\bibfnamefont{A.}~\bibnamefont{Soeda}},
  \bibinfo{author}{\bibfnamefont{P.~S.} \bibnamefont{Turner}},
  \bibnamefont{and} \bibinfo{author}{\bibfnamefont{M.}~\bibnamefont{Murao}}
  (\bibinfo{year}{2010}), \eprint{arXiv:1008.1128}.

\bibitem[{\citenamefont{Cirac et~al.}(2001)\citenamefont{Cirac, D\"ur, Kraus,
  and Lewenstein}}]{PhysRevLett.86.544}
\bibinfo{author}{\bibfnamefont{J.~I.} \bibnamefont{Cirac}},
  \bibinfo{author}{\bibfnamefont{W.}~\bibnamefont{D\"ur}},
  \bibinfo{author}{\bibfnamefont{B.}~\bibnamefont{Kraus}}, \bibnamefont{and}
  \bibinfo{author}{\bibfnamefont{M.}~\bibnamefont{Lewenstein}},
  \bibinfo{journal}{Phys. Rev. Lett.} \textbf{\bibinfo{volume}{86}},
  \bibinfo{pages}{544} (\bibinfo{year}{2001}).

\bibitem[{\citenamefont{\.Zyczkowski and Bengtsson}(2004)}]{ZcBn04}
\bibinfo{author}{\bibfnamefont{K.}~\bibnamefont{\.Zyczkowski}}
  \bibnamefont{and}
  \bibinfo{author}{\bibfnamefont{I.}~\bibnamefont{Bengtsson}},
  \bibinfo{journal}{Open Syst. Inf. Dyn.} \textbf{\bibinfo{volume}{11}},
  \bibinfo{pages}{3} (\bibinfo{year}{2004}).

\bibitem[{\citenamefont{Bengtsson and \.Zyczkowski}(2006)}]{BnZc06}
\bibinfo{author}{\bibfnamefont{I.}~\bibnamefont{Bengtsson}} \bibnamefont{and}
  \bibinfo{author}{\bibfnamefont{K.}~\bibnamefont{\.Zyczkowski}},
  \emph{\bibinfo{title}{Geometry of Quantum States}}
  (\bibinfo{publisher}{Cambridge University Press},
  \bibinfo{address}{Cambridge}, \bibinfo{year}{2006}).

\bibitem[{\citenamefont{Griffiths et~al.}(2006)\citenamefont{Griffiths, Wu, Yu,
  and Cohen}}]{PhysRevA.73.052309}
\bibinfo{author}{\bibfnamefont{R.~B.} \bibnamefont{Griffiths}},
  \bibinfo{author}{\bibfnamefont{S.}~\bibnamefont{Wu}},
  \bibinfo{author}{\bibfnamefont{L.}~\bibnamefont{Yu}}, \bibnamefont{and}
  \bibinfo{author}{\bibfnamefont{S.~M.} \bibnamefont{Cohen}},
  \bibinfo{journal}{Phys. Rev. A} \textbf{\bibinfo{volume}{73}},
  \bibinfo{pages}{052309} (\bibinfo{year}{2006}).

\bibitem[{\citenamefont{Stedman}(2009)}]{Stedman200909}
\bibinfo{author}{\bibfnamefont{G.~E.} \bibnamefont{Stedman}},
  \emph{\bibinfo{title}{Diagram Techniques in Group Theory}}
  (\bibinfo{publisher}{Cambridge University Press}, \bibinfo{year}{2009}),
  \bibinfo{edition}{1st} ed., ISBN \bibinfo{isbn}{9780521119702}.

\bibitem[{\citenamefont{Cvitanovic}(2008)}]{Cvitanovic200807}
\bibinfo{author}{\bibfnamefont{P.}~\bibnamefont{Cvitanovic}},
  \emph{\bibinfo{title}{Group Theory: Birdtracks, Lie's, and Exceptional
  Groups}} (\bibinfo{publisher}{Princeton University Press},
  \bibinfo{year}{2008}), ISBN \bibinfo{isbn}{9780691118369}.

\bibitem[{\citenamefont{Horodecki et~al.}(2009)\citenamefont{Horodecki,
  Horodecki, Horodecki, and Horodecki}}]{RevModPhys.81.865}
\bibinfo{author}{\bibfnamefont{R.}~\bibnamefont{Horodecki}},
  \bibinfo{author}{\bibfnamefont{P.}~\bibnamefont{Horodecki}},
  \bibinfo{author}{\bibfnamefont{M.}~\bibnamefont{Horodecki}},
  \bibnamefont{and}
  \bibinfo{author}{\bibfnamefont{K.}~\bibnamefont{Horodecki}},
  \bibinfo{journal}{Rev. Mod. Phys.} \textbf{\bibinfo{volume}{81}},
  \bibinfo{pages}{865} (\bibinfo{year}{2009}).

\bibitem[{\citenamefont{Tyson}(2003)}]{0305-4470-36-39-309}
\bibinfo{author}{\bibfnamefont{J.~E.} \bibnamefont{Tyson}},
  \bibinfo{journal}{Journal of Physics A: Mathematical and General}
  \textbf{\bibinfo{volume}{36}}, \bibinfo{pages}{10101} (\bibinfo{year}{2003}).

\bibitem[{\citenamefont{Nielsen and Chuang}(2000)}]{NielsenChuang200009}
\bibinfo{author}{\bibfnamefont{M.~A.} \bibnamefont{Nielsen}} \bibnamefont{and}
  \bibinfo{author}{\bibfnamefont{I.~L.} \bibnamefont{Chuang}},
  \emph{\bibinfo{title}{Quantum Computation and Quantum Information}}
  (\bibinfo{publisher}{Cambridge University Press}, \bibinfo{year}{2000}),
  \bibinfo{edition}{1st} ed., ISBN \bibinfo{isbn}{9780521635035}.

\bibitem[{\citenamefont{Nielsen}(1999)}]{PhysRevLett.83.436}
\bibinfo{author}{\bibfnamefont{M.~A.} \bibnamefont{Nielsen}},
  \bibinfo{journal}{Phys. Rev. Lett.} \textbf{\bibinfo{volume}{83}},
  \bibinfo{pages}{436} (\bibinfo{year}{1999}).

\bibitem[{\citenamefont{Leifer et~al.}(2003)\citenamefont{Leifer, Henderson,
  and Linden}}]{PhysRevA.67.012306}
\bibinfo{author}{\bibfnamefont{M.~S.} \bibnamefont{Leifer}},
  \bibinfo{author}{\bibfnamefont{L.}~\bibnamefont{Henderson}},
  \bibnamefont{and} \bibinfo{author}{\bibfnamefont{N.}~\bibnamefont{Linden}},
  \bibinfo{journal}{Phys. Rev. A} \textbf{\bibinfo{volume}{67}},
  \bibinfo{pages}{012306} (\bibinfo{year}{2003}).

\bibitem[{\citenamefont{Stein et~al.}(2010)}]{sage}
\bibinfo{author}{\bibfnamefont{W.}~\bibnamefont{Stein}} \bibnamefont{et~al.},
  \emph{\bibinfo{title}{{S}age {M}athematics {S}oftware ({V}ersion 4.6.1)}},
  \bibinfo{organization}{The Sage Development Team} (\bibinfo{year}{2010}),
  \bibinfo{note}{{\tt http://www.sagemath.org}}.

\end{thebibliography}
\end{document}